\documentclass[journal]{IEEEtran}
\usepackage[utf8]{inputenc}
\usepackage{graphicx}
\usepackage{amsmath}
\usepackage{amsfonts}
\usepackage{commath}
\usepackage{xcolor}
\usepackage{cite}
\usepackage{comment}
\usepackage{rotating}
\usepackage{multirow}
\usepackage{amssymb}
\usepackage{eqparbox}
\usepackage{algorithm}
\usepackage{algorithmic}
\usepackage{balance} 
\usepackage{enumerate}
\usepackage{tikz}
\usetikzlibrary{shapes,arrows,positioning,calc}
\tikzset{
block/.style = {draw, fill=white, rectangle, minimum height=2.5em, minimum width=3em},
tmp/.style = {coordinate},
sum/.style= {draw, fill=white, circle, node distance=1cm},
input/.style = {coordinate},
output/.style= {coordinate},
pinstyle/.style = {pin edge={to-,thin,black}}}
\usepackage{float}
\usepackage{amsthm}

\newtheorem{theorem}{Theorem}
\newtheorem{lemma}{Lemma}%[theorem]
\newtheorem{proposition}{Proposition}%[theorem]
\title{\huge A Reference Governor for Overshoot Mitigation of Tracking Control Systems}
\author{\IEEEauthorblockN{C. Freiheit%\IEEEauthorrefmark{1}
, D. M. Anand%\IEEEauthorrefmark{2}
, H. R. Ossareh}%\vspace{-0.3cm}}%\IEEEauthorrefmark{1}}\\
\thanks{C. Freiheit and H. R. Ossareh are with the University of Vermont, Burlington,
VT, 05405 USA, e-mail: \{collin.freiheit, hamid.ossareh\}@uvm.edu}% <-this % stops a space
\thanks{D. M. Anand is with the National Institute of Standards and Technology, USA, e-mail: dhananjay.anand@nist.gov}% <-this % stops a space
}% <-this % stops a space
\date{August 2019}
\begin{document}
%\begin{multicols}{2}
\maketitle

\begin{abstract}
This paper presents a novel reference governor scheme for overshoot mitigation in tracking control systems. Our proposed scheme, referred to as the Reference Governor with Dynamic Constraint (RG-DC), recasts the overshoot mitigation problem as a constraint management problem. The outcome of this reformulation is a dynamic Maximal Admissible Set (MAS), which varies in real-time as a function of the reference signal and the tracking output. The RG-DC employs the dynamic MAS to modify the reference signal to mitigate or, if possible, prevent overshoot. 
We present several properties of the dynamic MAS and the algorithms required to compute it. We also investigate the stability and recursive feasibility of the RG-DC, and present an interesting property of RG-DC regarding its effect on the governed system's frequency response. Simulation results demonstrate the efficacy of the approach, and also highlight its limitations.  This paper serves as an extension of our earlier paper on this topic.
\end{abstract}
%\vspace{-0.3cm}
\section{Introduction}
\label{motivation}
Overshoot in closed-loop control systems is often an undesired phenomenon. For example, position overshoot in servo controlled robots may result in collisions, and in regulated electronic power converters, overshoot may cause  overload currents. Surprisingly,  there are very few methods available in the literature of control systems dedicated to overshoot mitigation.  One obvious solution is feedforward plant inversion \cite{Devasia_2002,Zou_Devasia_2007}, wherein a pre-filter is used to eliminate the overshoot resulting from the underdamped and/or zero dynamics of the closed-loop system. However this strategy requires an exact model of the plant, which is not always available. Additionally, a stable non-minimum phase system poses the problem of system destabilization upon plant inversion. Another strategy is to use a detuned or a more complex controller within the loop; however, this approach has the downside of slowing down the system, increasing its complexity, or not being able to handle variability in the plant dynamics. Furthermore, this approach may not be applicable to off-the-shelf products or systems with legacy controllers. Other overshoot mitigation solutions in the literature include a cascade control scheme coupled with a sliding mode controller \cite{Sliding_Mode_2007}, and a feedback gain design method based on quantifier elimination \cite{quantifier_elimination_2015}. These solutions either require an accurate model of the plant or  increase the complexity of the inner loop. In this paper, we propose a novel overshoot mitigation strategy using the Reference Governor (RG) framework. Unlike the existing methods in the literature, the proposed strategy does not require modifications to the controller within the closed-loop system, does not require model inversion, and can be made robust to modeling errors.

To provide a brief background, RG \cite{Kolmanovsky_2014,Gilbert_Kolmanovsky_1995,Garone_2017,Gilbert_1999,Ossareh_2019,Osorio_2018,laracy2020constraint, osorio2019reference,
liu2018decoupled}  is a predictive control strategy that, similar to Model Predictive Control (MPC) \cite{Maciejowski_2008}, employs a prediction of the evolution of the system state to enforce pre-specified constraints on the inputs, states, or the outputs. Unlike MPC, however, RG modifies the reference to a pre-stabilized closed-loop control system and is primarily intended for constraint management. Moreover, RG is more numerically efficient than MPC, which makes it attractive for real-time control of fast processes. A block diagram of a closed-loop system controlled by a RG is depicted in Fig.~\ref{fig:RGblock}. RG employs the so-called Maximal Admissible Set (MAS) \cite{Gilbert_1991}, which characterizes the set of all initial conditions and inputs that satisfy the constraints for all time. The MAS is computed offline, allowing the RG to enforce the constraints in real-time by solving a linear program subject to  state and input values belonging to the MAS.

%The calculation of the MAS is performed offline. In real-time, RG enforces the constraints by solving a linear program subject to the state and input belonging to the MAS.

\begin{figure}
\centering
\begin{tikzpicture}[auto, node distance=1.5cm,>=latex']
\node [input, name=rinput] (rinput) {};
\node [block, right of=rinput,text width=1.5cm,align=center] (controller) {{\footnotesize Reference Governor}};
\node [block, right of=controller,node distance=3cm,text width=1.6cm,align=center] (system)
{{\footnotesize Closed-Loop System}};
\node [output, right of=system, node distance=2cm] (output) {};
\node [tmp, below of=controller,node distance=0.9cm] (tmp1){$s$};
\draw [->] (rinput) -- node{\hspace{-0.4cm}$r(t)$} (controller);
% \draw [->] (sum1) -- (controller);
\draw [->] (controller) -- node [name=v]{$v(t)$}(system);
\draw [->] (system) -- node [name=y] {$y(t)$}(output);
\draw [->] (system) |- (tmp1)-| node[pos=0.75] {$x(t)$} (controller);
% \draw [->] (extra)--(sum2);
% \draw [->] ($(0,1.5cm)+(extra)$)node[above]{$d_{\beta 2}$} -- (extra);
\end{tikzpicture}
\caption{Reference governor block diagram.} \label{fig:RGblock}
%\vspace{-0.3cm}
 \label{fig: Governor scheme}
\end{figure}
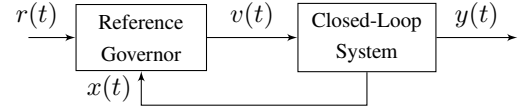

Traditional RG theory can only handle static constraints (i.e., constraints that do not vary with time). In this paper, we cast overshoot as a dynamic (i.e., time-varying) constraint on the tracking output (denoted by $y_{tr}$ hereafter) of the system.  Specifically,  if $y_{tr}(t)$ is above $r(t)$, we wish to maintain $y_{tr}$ above $r(t)$ for all future time. Similarly, if $y_{tr}(t)$ is below $r(t)$, we wish to hold $y_{tr}$ below $r(t)$ for all future time.
To accomplish this in the framework of the reference governor, we define the constraint set $\{y: y \leq r(t)\}$ whenever $y_{tr}(t)\leq r(t)$, and by the set $\{y:y \geq r(t)\}$ whenever $y_{tr}(t) > r(t)$. This dynamically-varying constraint leads to a novel, dynamically-varying MAS. We  present a unique modification of the RG theory to allow it to handle such dynamic MAS. We call this RG solution the Reference Governor with Dynamic Constraint (RG-DC).

The dynamic nature of the MAS and our RG-DC formulation raise the following questions:

\begin{enumerate}
    \item Does the number of inequalities required to describe the dynamic MAS change as the reference varies in real-time? %Does the admissibility index of the MAS, defined as the number of inequalities required to describe it, change if the constraint varies in real-time? 
    \item What is the geometric and algebraic relationship between the instances of the dynamic MAS at different times?
    \item Does the RG-DC guarantee constraint satisfaction  for all time?
    \item Can the RG-DC destabilize the control loop?
    \item %What is the extent of the additional computational complexity of the RG-DC compared to the standard RG?
    How much additional computational complexity does the RG-DC introduce compared to the standard RG?
    \item How can the RG-DC be made robust to model uncertainty and unknown disturbances?    \label{list}
\end{enumerate}
All of these questions will be addressed in this paper. To summarize, the original contributions of this paper are a new approach (RG-DC) to mitigate overshoot in closed-loop control systems, and the analysis and demonstration of the six questions raised above.  Additionally, we reveal an interesting property of RG-DC regarding its effect on the governed system's frequency response. Specifically,  the RG-DC can act as a novel nonlinear filter to eliminate resonance in closed-loop systems caused by underdamped poles and/or zero dynamics.

Note that reference \cite{Kalabic_Kolmanovsky_2014} investigates a RG solution for systems with  slowly-varying constraints.  However, the results of \cite{Kalabic_Kolmanovsky_2014} are not applicable to our problem because the dynamic constraint considered in our paper may vary rapidly. Furthermore, similar to our paper, reference \cite{TURBO} briefly considers overshoot mitigation in the framework of RG, but it does not provide a rigorous answer to the questions raised above. 

This paper is an extended version of our earlier IEEE L-CSS publication \cite{Freiheit_2020}. Furthermore, this paper corrects a small error in the L-CSS publication which is explained in the footnote of page $3$.

\section{Review of Reference Governors}
\label{review}
Consider Fig. \ref{fig:RGblock}, in which the ``closed-loop system" is described by the single-input multi-output  discrete-time, stable linear system: 
\begin{equation*}
\begin{aligned}
    x(t+1) = Ax(t)+Bv(t)\\
    y(t) = Cx(t)+Dv(t)
\end{aligned}
\end{equation*}
where the output $y$ is subject to the following polyhedral constraints:
%where the output $y$ is to be constrained as follows:
\begin{equation}\label{eqn:Sys}
    y(t) \in \mathbb{Y} \triangleq \{y: S y\leq s\}
\end{equation}
Vector inequalities here and throughout the paper are to be interpreted element-wise. In general, the set in \eqref{eqn:Sys} may be unbounded. The RG employs the so-called maximal admissible set (MAS), denoted by $O_\infty$, which is the set of all states and control inputs that satisfy  \eqref{eqn:Sys} for all time:
\begin{equation}\label{eq:Oinfdefintiion}
O_\infty = \big\{(x,v): x(0)=x,\ v(t)=v,\ y(t) \in \mathbb{Y},\ \forall t\in \mathbb{Z}^+\big\}
\end{equation}
 As seen in \eqref{eq:Oinfdefintiion}, to construct MAS,  $v(t) = v$ is held constant for all $t$. Using this assumption, the evolution of the output $y(t)$ can be expressed explicitly as a function of $x(0)=x$ and $v$:
\begin{equation}
    y(t) = CA^t x + \left(C(I-A^t)(I-A)^{-1}B+D\right)v
    \label{eqn:y[t]}
\end{equation}

\noindent Therefore, MAS in \eqref{eq:Oinfdefintiion} can be characterized by a polyhedron defined by an infinite number of inequalities:
{\small
\begin{equation} \label{eqn:Sy[t]}
\begin{aligned}
    &O_\infty = \Big\{(x,v):\\
    &SCA^t x + 
    S\left(C(I-A^t)(I-A)^{-1}B+D\right)v \leq s,\ \forall t\in \mathbb{Z}^+\Big\} 
    \end{aligned}
\end{equation}}\par
It is shown in \cite{Gilbert_1991}  that,  under mild assumptions on $C$ and $A$, it is possible to make this set finitely determined (i.e., be described by a finite number of inequalities) by constraining the steady-state value of $y$, denoted by $y(\infty)$, to the interior of the constraint set:
\begin{equation} \label{eqn:Sy[infty]}
y(\infty) \triangleq  \left(C(I-A)^{-1}B+D\right)v \in (1-\epsilon) \mathbb{Y}
\end{equation}

where $\epsilon \in \mathbb{R}^+$ is a small number. As shown in \cite{Gilbert_1991}, after introducing  \eqref{eqn:Sy[infty]} in the MAS, there exists a finite prediction time $j^*$, where the inequalities corresponding to all future prediction times ($t>j^*$) are redundant. The smallest such $j^*$ is referred to as the {\it admissibility index} of the MAS. 

Combining \eqref{eqn:Sy[t]} and \eqref{eqn:Sy[infty]}, we obtain an inner approximation of $O_\infty$, denoted by $\widetilde{O}_{\infty}$, which can be represented by:
\begin{equation}
    \widetilde{O}_{\infty} = \Big\{(x,v):H_x x + H_v v \leq h\Big\}
    \label{eqn:O_infty}
\end{equation}
where the matrices $H_x$, $H_v$, and $h$ are finite dimensional.  Note that $h$ is a vector with all elements equal to $s$, except the first block of rows, which is $(1-\epsilon)s$. To numerically construct $H_x,\ H_v$ and $h$, we begin with the steady-state inequality in \eqref{eqn:Sy[infty]} and iteratively add the inequalities in \eqref{eqn:Sy[t]} starting with $t=0$. After each $t$, we check if the newly added rows are all redundant. If this is so, $j^*$ has been reached and the construction of $\widetilde{O}_\infty$ is complete. 
 
We now review the algorithm provided in  \cite{Gilbert_1991} to check for redundancy. This algorithm is leveraged in Section~\ref{DCM} for the analysis of our dynamic MAS. Given any polyhedron defined by $Mz\leq N$ and a scalar inequality given by $c^T z\leq d$, to determine if the inequality is redundant with respect to the polyhedron, it is common practice to solve the following linear program (LP) \cite{Gilbert_1991}:
\begin{equation}
  % \begin{array}{c}
    f=\max\ c^T z \quad \mathrm{subject\ to} \quad  Mz\leq N
   % \end{array} 
    \label{eqn:RG_lp}
\end{equation}
If $f \leq d$, the new inequality is redundant. To apply this idea to  MAS, suppose MAS has been partially constructed with the inequalities in
\eqref{eqn:Sy[infty]} and  \eqref{eqn:Sy[t]} from $t=0$ up to $t=j$, for some $j$. Let $H_x$, $H_v$, $h$ represent the matrices of this partially constructed MAS. We wish to test whether an  inequality in \eqref{eqn:Sy[t]} with $t=j+1$ is redundant. The LP above can be used for this purpose, with $M=[H_x, H_v]$, $N=h$, $z=(x,v)$, and $c^T$ and $d$ representing the  inequality being tested for redundancy.

The final step in the RG is to select an optimal control input that will not cause a constraint violation. %Assuming that the current state, $x(t)$, and the previous control command, $v(t-1)$, belong to $O_{\infty}$, then we can formulate the  following
The RG update law that achieves this is:
\begin{equation}
    v(t) = v(t-1)+\kappa \left(r(t)-v(t-1)\right)
    \label{eqn:v[t]}
\end{equation}
where $\kappa \in [0,1]$. To select $\kappa$, we solve the following linear program: 

\begin{equation}\label{eq:kappa_for_RG}
\begin{aligned}
&\underset{\kappa\in [0,1]}{\text{maximize}}
& & \mathrm{\kappa} \\
& \hspace{10pt} \text{s.t.}
& & v(t)=v(t-1)+\kappa\left(r(t)-v(t-1)\right)\\
&&&\left(x(t),\ v(t)\right) \in \widetilde{O}_{\infty}
\end{aligned}
\end{equation}
where $x(t)$, $r(t)$, and $v(t-1)$ are known parameters at time $t$. If $\kappa=0$, the control command from the previous timestep is maintained to avoid constraint violation, and if $\kappa = 1$, the reference $r(t)$ is feasible and, therefore, $v(t) = r(t)$.

\section{Reference Governor with Dynamic Constraint (RG-DC)}
\label{DCM}  
Consider the asymptotically stable system
\begin{equation}
    \begin{aligned}
    x(t+1) &= Ax(t)+Bv(t)\\
    y_{tr}(t) &= C_{tr}x(t)+D_{tr}v(t)\\
    y_{st}(t) &= C_{st}x(t)+D_{st}v(t)
    \end{aligned}
    \label{eqn:RGDC_system}
\end{equation}
with DC gain from $v$ to $y_{tr}$ equal to $1$, where $y_{tr}\in\mathbb{R}$ is the tracking output on which we wish to enforce the dynamic overshoot constraint (as explained below). Additionally, $y_{st}\in \mathbb{R}^p$ refers to constrained outputs, on which we wish to enforce standard static constraints:
\begin{equation}
    \begin{aligned}
    y_{st}(t)\in \mathbb{Y}_{st}\triangleq\{y:S_{st}y\leq s_{st}\}
    \end{aligned}
    \label{eqn:static_set}
\end{equation}
It should be noted that, because $y_{tr}$ is the output of the plant within the closed-loop system, there is no feedforward from $v$ to $y_{tr}$ in practice. Thus, for the remainder of the paper, we assume that $D_{tr}=0$. Note however that $D_{st}$ is allowed to be non-zero because static constraints could be imposed on controller states or the controller output, which may require feedthrough.

For overshoot mitigation, we impose that $y_{tr}$ be constrained by  the reference $r$, which may vary with time. To do so, two cases must be considered: the first case is where $y_{tr}(t) \leq r(t)$ at the current time $t$, for which we define overshoot by the following condition:
$\exists k > t$ such that $y_{tr}(k) > r(t)$. Thus, to prevent overshoot, we must enforce the following constraint: $y_{tr}(k) \in \{y:y\leq r(t)\}$ for all $k > t$. In the second case, $y_{tr}(t) > r(t)$ at the current time $t$, for which we define overshoot by 
$\exists k > t$ such that $y_{tr}(k) < r(t)$ and the constraint by $y_{tr}(k) \in \{y:y\geq r(t)\}$ for all $k > t$. Note that we have chosen the constraint sets to be closed (i.e., the inequalities are not strict), which is necessary to ensure that the linear programs that arise in RG-DC are well-posed. The above leads to a time-varying constraint set that depends on both $y_{tr}(t)$ and $r(t)$:
\begin{equation}
    \begin{aligned}
     \mathbb{Y}_{tr}(r(t),\  y_{tr}(t))\triangleq  \begin{cases} 
      \{y:y\leq r(t)\} & y_{tr}(t) \leq r(t) \\
      \{y:y\geq r(t)\} & y_{tr}(t) > r(t) \\
%      \mathrm{no\ constraints\ imposed} & y(t) = r(t) 
   \end{cases}
    \end{aligned}
    \label{eqn:dynamic_set}
\end{equation}
The goal is to enforce $y_{tr}(k) \in \mathbb{Y}_{tr}(r(t),\  y_{tr}(t))$ for all $k > t$.

We now define the maximal admissible sets for this system. For the static constraint in \eqref{eqn:static_set}, we create MAS as discussed previously in Section \ref{review}. We denote this MAS by $O_{\infty,st}$. For the dynamic MAS, note that the second constraint in \eqref{eqn:dynamic_set} can be re-written as $\{y: -y \leq -r(t)\}$, which implies that both constraints in \eqref{eqn:dynamic_set} can be cast in the form \eqref{eqn:Sys}, where $S$ takes on the values of $1$ or $-1$ and $s$ takes on the values of $r(t)$ or $-r(t)$. Therefore, the definition of MAS remains the same as \eqref{eq:Oinfdefintiion}, with the exception that, since $\mathbb{Y}_{tr}$ depends on $r(t)$ and $y_{tr}(t)$, so does the MAS. We thus denote this dynamic MAS by $O_{\infty,tr}\left(r(t),y_{tr}(t)\right)$. In Subsection \ref{MAS}, we analyze the properties and computation of this dynamic MAS.

The proposed reference governor scheme (RG-DC) employs the intersection of the static MAS (for constraints on $y_{st}$) and the dynamic MAS (for constraints on $y_{tr}$) to compute $\kappa$ from \eqref{eq:kappa_for_RG} and $v(t)$ from  \eqref{eqn:v[t]}. We will discuss  the stability and recursive feasibility of the system with RG-DC, as well as the implementation aspects, in Subsection \ref{RG}. We also discuss a robust formulation of RG-DC to handle plant-model mismatch and unknown disturbances.

For simplicity, we assume that all states of the system are available for feedback. If not, a set-based observer can be designed as is done in \cite{kalabic2015reference}.

% \vspace{-0.2cm}
\subsection{Computational aspects and properties of the dynamic MAS}
\label{MAS}

We first address the computation of the dynamic MAS defined above (the computation of the static MAS is standard and will not be addressed). For this investigation, we seek to develop a polyhedral characterization of the dynamic MAS, parameterized on ${r(t)=r}$ and $y_{tr}(t)$.

First suppose  that $r>0$ denoted $r^+$. We will relax this assumption later. Now consider the inequalities in \eqref{eqn:Sy[t]}. Recall from above that $S$ takes on the value of $1$ (in which case $s=r^+$), or $-1$ (in which case $s=-r^+$). If $S=1$, the steady-state halfspace should be shrunk to: $v\leq (1-\epsilon)r^+$ and the inequalities in \eqref{eqn:Sy[t]} become:
\begin{equation}
C_{tr}A^t x +    C_{tr}(I-A^t)(I-A)^{-1}B v \leq r^+
\label{eqn:less}
\end{equation}
If $S=-1$, the steady-state halfspace should be shrunk to: $v\geq (1+\epsilon)r^+$ and the inequalities in \eqref{eqn:Sy[t]} become:
\begin{equation}
C_{tr}A^t x +    C_{tr}(I-A^t)(I-A)^{-1}B v \geq r^+
\label{eqn:greater}
\end{equation}
A  polyhedral representation of MAS constructed from the tightened steady-state constraint  $v \leq (1-\epsilon) r^+$ and the inequalities in \eqref{eqn:less} for all $t\geq 0$ is given by:
\begin{equation}
  O_{\infty}^-(r^+)=\Big\{(x,v):H_x x + H_v v \leq r^+ h^- \Big\} 
  \label{eqn:oinf-}
\end{equation}
Similarly, a representation of MAS  using \eqref{eqn:greater} with the tightened steady-state constraint  $v \geq (1+\epsilon) r^+$ is:
\begin{equation}
 O_{\infty}^+(r^+)=\Big\{(x,v):H_x x + H_v v \geq r^+ h^+ \Big\}
 \label{eqn:oinf+}
\end{equation}
where $h^{-}$ and $h^{+}$ are vectors of all $1$s except the first block of rows, which are $1-\epsilon$ and $1+\epsilon$, respectively\footnote{As mentioned in the Introduction, this paper corrects a small error in the L-CSS publication \cite{Freiheit_2020} involving the steady-state halfspaces of the dynamic MASs. In the L-CSS publication, the steady state halfspace constraint for all dynamic MASs was $(1-\epsilon)r$. This is problematic because $O^+_\infty(r)$ and $O^-_\infty(r)$ from Cases $2$ and $3$ from Table \ref{tab:4_combinations} of \cite{Freiheit_2020} would not necessarily be finitely determined because the steady-state halfspaces are not being shrunk. In this paper we fix this problem by letting the steady-state halfspaces of $O^+_\infty(r)$ and $O^-_\infty(r)$ from Cases $2$ and $3$ from Table \ref{tab:4_combinations} of \cite{Freiheit_2020} have constraints of $(1+\epsilon)r$, making the sets finitely determined. Note that all of the results (Lemmas, Propositions, Theorem, and example) of the L-CSS publication still hold, we simply modify the notation in this paper to distinguish between the cases of $(1-\epsilon)r$ and $(1+\epsilon)r$ for the steady-state halfspaces of the dynamic MAS.}. Note that in order to explicitly show the dependence of the sets on $r^+$, we have formulated \eqref{eqn:oinf-}-\eqref{eqn:oinf+} with $r^+h^-$ and $r^+h^+$ on the right hand sides (instead of simply $h$ as in \eqref{eqn:O_infty}). For now, we consider $H_x, H_v, h^-,h^+$ as being infinite dimensional matrices (i.e., a redundancy check was not performed when forming $O_{\infty}^-$ and $O_{\infty}^+$). Since we know, from Section \ref{review}, that both \eqref{eqn:oinf-} and \eqref{eqn:oinf+} must be finitely determined for a fixed $r^+$, our goal now is to study the admissibility index of these sets as functions of $r^+$. 

Recall from Section \ref{review} that to find the admissibility index of a MAS, we construct it row by row and stop when redundancy is detected. Furthermore, to detect redundancy, we use the linear program (LP) in \eqref{eqn:RG_lp}. While redundancy can be checked for  $O_{\infty}^-(r^+)$ using the same approach, $O_{\infty}^+(r^+)$ requires a LP of a different form. To formulate a LP for $O_{\infty}^+(r^+)$, we represent \eqref{eqn:oinf+} in the form of \eqref{eqn:oinf-}, yielding $-(H_x x + H_v v) \leq -r^+ h^+$. Upon applying \eqref{eqn:RG_lp} to this inequality and simplifying the resulting LP, we obtain the following adaptation of \eqref{eqn:RG_lp}:

\begin{equation}
  % \begin{array}{c}
    f=\min\ c^T z \quad \mathrm{subject\ to} \quad  Mz\geq N
 %   \end{array} 
    \label{eqn:RG_lp-}
\end{equation}

 To proceed with our analysis of admissibility index, we first show, with support of Lemma \ref{gamma}, that the individual admissibility indices of $O_{\infty}^-(r^+)$ and $O_{\infty}^+(r^+)$ are unchanged for any $r^+$. %Theorem \ref{gamma} 
\begin{lemma}
\label{gamma}
Suppose the unique maximizer of
\begin{equation}
\begin{aligned}
    \mathrm{max}\ c^{T}z \quad 
\mathrm{subject\ to} \quad Mz\leq N
\end{aligned}
\label{eqn:lp_unscaled}
\end{equation}

\noindent is given by $z^*$. Then, for any $\gamma \in \mathbb{R^+}$, the maximizer of
\begin{equation}
\begin{aligned}
    \mathrm{max}\ c^{T}z
    \quad 
\mathrm{subject\ to} \quad
    Mz\leq\gamma N
\end{aligned}
\label{eqn:lp}
\end{equation}

\noindent is given by $\gamma z^*$. Furthermore, the optimal values of the objective functions in  \eqref{eqn:lp_unscaled} and \eqref{eqn:lp} are given by $c^T z^*$ and $c^T\gamma z^*$. That is, the optimal value of \eqref{eqn:lp} is $\gamma$ times larger than that of \eqref{eqn:lp_unscaled}.
\end{lemma}
\begin{proof}
Given \eqref{eqn:lp_unscaled}, we rewrite the constraint by multiplying both sides by $\gamma$: $M (\gamma z) \leq \gamma N$. Furthermore, we can multiply the cost function by $\gamma$, which is permitted  because a positive scaling on the objective function of a linear programming problem does not change the optimizer. We thus obtain the equivalent linear program:
\begin{equation*}
\begin{aligned}
    \mathrm{max}\ c^{T} (\gamma z) \quad 
\mathrm{subject\ to} \quad M (\gamma z)\leq \gamma N
\end{aligned}
\end{equation*}
which has the same optimizer as \eqref{eqn:lp_unscaled} but a different objective function value. Finally, we can perform a change of variable $\gamma z \rightarrow z$ to transform this optimization into \eqref{eqn:lp}. It can be concluded that, if the optimizer of \eqref{eqn:lp_unscaled} is $z^*$, the optimizer of \eqref{eqn:lp} must be $\gamma z^*$. 

\end{proof}

Noting that \eqref{eqn:RG_lp} and \eqref{eqn:lp_unscaled} are the same linear program, we can now apply the results of Lemma \ref{gamma} to the LP in \eqref{eqn:RG_lp} to show that admissibility index of $O_{\infty}^-(r^+)$ in \eqref{eqn:oinf-} is unaffected by a positive scaling on $r^+$ (the same argument holds true for $O_{\infty}^+(r^+)$ as well). Specifically, suppose the redundancy of a new inequality $c^Tz\leq r^+$ is tested against the partially constructed MAS given by $Mz\leq N$,  where $M=[H_x,\ H_v]$, $N=r^+h^-$, and $z=(x,\ v)$. 
From Lemma \ref{gamma}, scaling $r^+$ by $\gamma \in \mathbb{R}^+$ (i.e., replacing $r^+$ with $\gamma r^+$) scales the optimal solution of the LP by $\gamma$. However, the constraint being tested for redundancy is also scaled by $\gamma$, which implies that the redundancy of $c^Tz\leq \gamma r^+$ is unaffected by $\gamma$. Therefore, we conclude that the admissibility index of $O_\infty^-(r^+)$ is unaffected by a positive scaling on $r^+$.  %Note that, in general, \eqref{eqn:oinf-} and \eqref{eqn:oinf+} each have their own admissibility indices.

Now assume $r<0$, denoted $r^-$. A  polyhedral representation of MAS for the case of $S=1$ is given by 
\begin{equation}
  O_{\infty}^-(r^-)=\Big\{(x,v):H_x x + H_v v \leq r^- h^+ \Big\} 
  \label{eqn:oinf--}
\end{equation}
Note that we choose $h^+$ in \eqref{eqn:oinf--} to ensure that the steady-state constraint is indeed contracted. Similarly, a representation of MAS for the case of $S=-1$ is:
\begin{equation}
 O_{\infty}^+(r^-)=\Big\{(x,v):H_x x + H_v v \geq r^- h^- \Big\}
 \label{eqn:oinf+-}
\end{equation}
Using Lemma \ref{gamma}, we can conclude that the admissibility index of $O_{\infty}^-(r^-)$ in \eqref{eqn:oinf--} is also unaffected by a positive scaling on $r^-$ and that the same argument holds true for $O_{\infty}^+(r^-)$ in \eqref{eqn:oinf+-}.

Now suppose that $r$ is allowed to be any non-zero real number. If $r$ changes sign, the geometric properties of the MAS change (graphical argument presented in Fig. \ref{fig:Ineqs}), which in turn changes the linear programs in \eqref{eqn:RG_lp} and \eqref{eqn:RG_lp-}.\begin{figure}
    \centering
    \includegraphics[width=\columnwidth]{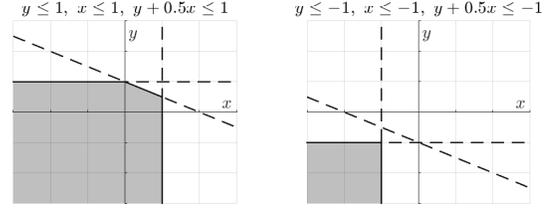}
    \caption{Two systems of linear inequalities in $\mathbb{R}^2$: constraints are positive in the left plot and negative in the right plot.}
    \label{fig:Ineqs}
\end{figure} However, with the help of  Lemma \ref{x-x} below, we prove that the admissibility index of \eqref{eqn:oinf-}, is equivalent to the admissibility index of \eqref{eqn:oinf+-}, and similarly that the admissibility index of \eqref{eqn:oinf+}, is equivalent to the admissibility index of \eqref{eqn:oinf--}. %We introduce a new notation for the dynamic constraint MASs after 
\begin{lemma}
\label{x-x}
suppose the maximizer of 
\begin{equation}
\begin{aligned}
    \mathrm{max}\ c^{T}z \quad
    \mathrm{subject\ to}\quad Mz\leq N
\end{aligned}
\label{eqn:lpe1}
\end{equation}
is given by $z^*$. Then, the minimizer of 
\begin{equation}
\begin{aligned}
    \mathrm{min}\ c^{T}z\quad
    \mathrm{subject\ to}\quad Mz\geq -N
\end{aligned}
\label{eqn:lpe2}
\end{equation}
is given by $-z^*$.
\end{lemma}
The proof for Lemma \ref{x-x} is simple and follows similarly to Lemma \ref{gamma}, we thus omit the proof. From Lemma \ref{x-x}, we can conclude that if the maximum value of the objective function in  \eqref{eqn:lpe1} is $c^Tz^*$, then the minimum of \eqref{eqn:lpe2} is $c^T(-z^*)$. %Therefore, when redundancy is checked by comparing $f$ to the corresponding constraint which we denote $d$ and $-f$ to  the corresponding constraint $-d$, the admissibility index is the same no matter which of the two LPs are used. 
Applying this result to \eqref{eqn:oinf-}, we see that we will obtain the same admissibility index as for \eqref{eqn:oinf+-}. The same can be said between \eqref{eqn:oinf+} and \eqref{eqn:oinf--}.
 
 If we combine the information presented in Lemmas \ref{gamma} and \ref{x-x}, it follows that, for a given $r$, $O_\infty^+$ and $O_\infty^-$ can both be uniquely defined by only two representations of $H_x,\ H_v,$ and $h$, which is the novel result of this subsection. This is summarized in Table \ref{tab:4_combinations}, where superscripts $^+$ and $^-$ are used to denote the two possible representations. 
 This table also highlights the relationship between the dynamic MAS, $O_{\infty,tr} \left(r(t),\ y_{tr}(t)\right)$, and the sets $O_\infty^-(r^+)$, $O_\infty^+(r^+)$, $O_\infty^-(r^-)$, and $O_\infty^+(r^-)$ in \eqref{eqn:oinf-}, \eqref{eqn:oinf+}, \eqref{eqn:oinf--} , and \eqref{eqn:oinf+-} respectively.

\renewcommand{\arraystretch}{1.4}
\begin{table}[ht!]
%\vspace{-0.3cm}
    \setlength\extrarowheight{6pt}
    \begin{tabular}{c|c|c|}
        \multicolumn{1}{c}{} & \multicolumn{1}{c}{$\{y_{tr}\leq r\}$} & \multicolumn{1}{c}{$\{y_{tr}> r\}$}\\
        \cline{2-3}
        \multirow{2}{.5em}{\rotatebox{90}{\scalebox{.9}[1.0]{$r=r^+>0$}}} & \multirow{1}{12em}{Case 1: $O_{\infty,tr}=O_{\infty}^-(r^+)$} & \multirow{1}{12em}{Case 2: $O_{\infty,tr}=O_{\infty}^+(r^+)$} \\
        & \multirow{1}{12em}{Minimal representation:\\$\big\{H_{x}^-x + H_{v}^-v \leq r h^-\big\}$} & \multirow{1}{12em}{Minimal representation:\\$\big\{H_{x}^+x + H_{v}^+v \geq r h^+\big\}$}\\
        \cline{2-3}
        \multirow{2}{.5em}{\rotatebox{90}{\scalebox{.9}[1.0]{$r=r^-<0$}}} & \multirow{1}{12em}{Case 3: $O_{\infty,tr}=O_{\infty}^-(r^-)$} & \multirow{1}{12em}{Case 4: $O_{\infty,tr}=O_{\infty}^+(r^-)$} \\ 
        & \multirow{1}{12em}{Minimal representation:\\$\big\{H_{x}^+x + H_{v}^+v \leq r h^+\big\}$} & \multirow{1}{12em}{Minimal representation:\\ $\big\{H_{x}^-x + H_{v}^-v \geq r h^-\big\}$}\\
        \cline{2-3}
    \end{tabular}
    \vspace{.1in}
    \caption{The representation of the dynamic  maximal admissible set.}
    \label{tab:4_combinations}
    % \vspace{-0.3cm}
\end{table}
\renewcommand{\arraystretch}{1}

Note that the MASs in Cases $1$ and $4$ share the same matrices $H_x^-, H_v^-, h^-$. Similarly, the MASs in Cases $2$ and $3$ share the same  matrices $H_x^+, H_v^+, h^+$. Furthermore, these matrices are constant and  do not depend on the magnitude of $r$. Therefore, to construct these matrices, we can make simplifying assumptions on $r$. Specifically, to compute $H_x^-, H_v^-, h^-$, we can assume that $r=1$  and leverage the standard methods presented in Section \ref{review}. Similarly, to compute $H_x^+, H_v^+, h^+$, we can assume that $r=-1$ and use the standard methods.

Note that the only remaining case to consider is $r=0$. In this case, the constraint on the steady state cannot be shrunk (because $(1-\epsilon)r=(1+\epsilon)r=r$ when $r=0$), resulting in a MAS that may not necessarily be finitely determined. We resolve this by approximating it by the representation with the higher number of rows. %, similar to the previous paragraph. % for the calculation of $\kappa$. \color{black}%to get a better approximation.
This completes the answer to the first question raised in the Introduction around the admissibility index of the dynamic MAS.

We next study the geometric properties of the dynamic MAS as a function of $r$ using Propositions \ref{thm3} and \ref{subset} below.    %With the notion that a positive scaling on the $h$ vector in $O_\infty$ does not change the admissibility index (i.e. $H_x$ and $H_v$ are unchanged), we reveal an interesting property of the MAS. 
% prove that if $(x,v) \in \{H_x x + H_v v \leq h\}$, then $\gamma(x,v)\in \{H_x x + H_v v \leq \gamma h\}$.
\begin{proposition}
\label{thm3}
Let $O_\infty^-(r^+)$, $O_\infty^+(r^+)$, $O_\infty^-(r^-)$ and $O_\infty^+(r^-)$ be defined by \eqref{eqn:oinf-}, \eqref{eqn:oinf+}, \eqref{eqn:oinf--}, and \eqref{eqn:oinf+-} respectively; and let $r_1,\ r_2 \in \mathbb{R}\setminus \{0\}$. Then, the following holds.
\begin{enumerate}[i)]
    \item If $\frac{r_1}{r_2} > 0$, then $O_\infty^- (r_1) = \frac{r_1}{r_2} O_\infty^- (r_2)$ and $O_\infty^+ (r_1) = \frac{r_1}{r_2} O_\infty^+ (r_2)$
    \item If $\frac{r_1}{r_2} < 0$, then $O_\infty^- (r_1) = \frac{r_1}{r_2} O_\infty^+ (r_2)$ and $O_\infty^+ (r_1) = \frac{r_1}{r_2} O_\infty^- (r_2)$
\end{enumerate} 
%Let $O_\infty = \{(x,v):H_x x + H_v v \leq h\}$, and $O_{\infty}^* = \{(x,v):H_x x + H_v v \leq \gamma h\}$ then, $O_{\infty}^* = \gamma O_{\infty}$ % and $O_{\infty}\subseteq O_{\infty}^*$ if $\gamma\geq1$
%where $\gamma \in \mathbb{R}^+$
%$(x,v)\in O_\infty \Longleftrightarrow \gamma(x,v)\in O_{\infty *}$ where $\gamma \in \mathbb{R}^+$ 
\end{proposition}
\begin{proof}
For clarity throughout the proof, let superscripts $^+$ (positive) and $^-$ (negative) denote the signs of $r_1$ and $r_2$. We prove case $i$) and $ii$) for $O_\infty^-(r_1^+)$; the rest of the cases can be proven similarly.

$i$) %$O_\infty^-(r_1^+)$\\
Let $(x,v)\in O_\infty^- (r_1^+)$. Then, it follows from \eqref{eqn:oinf-} that
\begin{equation*}
\begin{aligned}
&H_x x + H_v v \leq r_1^+ h^- \quad \\
\Longleftrightarrow \quad
&\frac{r_2^+}{r_1^+} \left(H_x x + H_v v\right) \leq \frac{r_2^+}{r_1^+} \left(r_1^+ h^-\right)\\
\Longleftrightarrow \quad &H_x \left(\frac{r_2^+}{r_1^+} x\right) + H_v \left(\frac{r_2^+}{r_1^+} v\right) \leq r_2^+ h^- \quad\\ 
\Longleftrightarrow \quad
%(\gamma x,\gamma v)\in O_{\infty}^*\\
&\frac{r_2^+}{r_1^+} \left(x,v\right)\in O_{\infty}^-(r_2^+) \\
 \Longleftrightarrow \quad
 &(x,v)\in \frac{r_1^+}{r_2^+}O_{\infty}^-(r_2^+) 
\end{aligned}
\end{equation*}

$ii$) % $O_\infty^-(r_1^+)$\\
Let $(x,v)\in O_\infty^- (r_1^+)$. Then, it follows from \eqref{eqn:oinf-} that
\begin{equation*}
\begin{aligned}
&H_x x + H_v v \leq r_1^+ h^- \quad\\ \Longleftrightarrow \quad
&\frac{r_2^-}{r_1^+} \left(H_x x + H_v v\right) \geq \frac{r_2^-}{r_1^+} \left(r_1^+ h^-\right)\\
\Longleftrightarrow \quad &H_x \left(\frac{r_2^-}{r_1^+} x\right) + H_v \left(\frac{r_2^-}{r_1^+} v\right) \geq r_2^- h^- \quad\\ \Longleftrightarrow \quad
%(\gamma x,\gamma v)\in O_{\infty}^*\\
&\frac{r_2^-}{r_1^+} \left(x,v\right)\in O_{\infty}^+(r_2^-) \\
 \Longleftrightarrow \quad
 &(x,v)\in \frac{r_1^+}{r_2^-}O_{\infty}^+(r_2^-) 
\end{aligned}
\end{equation*}
The remaining $6$ cases $$O_\infty^-(r_1^-),\ O_\infty^+(r_1^+),\ O_\infty^+(r_1^-)\ \text{for\ both}\ i)\ \text{and}\ ii)$$ can be proven similarly. Furthermore, the reverse direction of each case can be proven.
\end{proof}
\noindent Proposition \ref{thm3} sheds light on the geometric relationship between $O_\infty^+$ and $O_\infty^-$. For example, for positive values of $r$, $O_\infty^-(r)$ is scaled radially from the origin as $r$ varies.

Another important result, which  ties into recursive feasibility of the RG-DC as addressed in Subsection \ref{RG}, is as follows. % If $0 \leq \frac{r_1}{r_2}\leq 1$, then $O_\infty^-(r_1)\subseteq O_\infty^-(r_2)$ and $O_\infty^+(r_1)\subseteq O_\infty^+(r_2)$. Similarly, if $\frac{r_1}{r_2}$, then $O_\infty^+(r_1)\subseteq O_\infty^+(r_2)$

\begin{proposition}
\label{subset}
Suppose $r_2\geq r_1$, then $O_\infty^-(r_1)\subseteq O_\infty^-(r_2)$. Similarly, if $r_2\leq r_1$, then $O_\infty^+(r_1)\subseteq O_\infty^+(r_2)$. 
%Let $(x,v) \in O_{\infty}(r_1)$ and let ${r_1}$ and ${r_2}$ be of the same sign, if $\abs{r_2}\geq\abs{r_1}$, then $(x,v) \in O_{\infty}(r_2)$ for Cases 1 and 4. Additionally, if $\abs{r_2}\leq\abs{r_1}$, then $(x,v) \in O_{\infty}(r_2)$ for Cases 2 and 3.
\end{proposition}
\begin{proof}
We prove the first statement of Proposition \ref{subset}. The second statement follows similarly.\\
Let $r_2\geq r_1$, and $(x,v)\in O_\infty^-(r_1)$. We first consider the case where $r_1\geq 0$, denoted $r_1^+$, and where $r_2\geq 0$, denoted $r_2^+$.  From \eqref{eqn:oinf-}, it is true that $H_x x + H_v v \leq r_1^+ h^-$. Therefore, because $r_2^+h^- \geq r_1^+h^-$, $H_x x + H_v v \leq r_2^+ h^-$. From here it can be concluded that any $(x,v) \in O_\infty^-(r_1^+)$ also belongs to $O_\infty^-(r_2^+)$. Therefore, $O_\infty^-(r_1^+)\subseteq O_\infty^-(r_2^+)$.
Now let us consider the case where $r_1\leq0$, denoted $r_1^-$, and where $r_2\geq0$, denoted $r_2^+$. Clearly, $r_2^+ \geq r_1^-$. Let $(x,v)\in O_\infty^-(r_1^-)$.  From \eqref{eqn:oinf--}, it is true that $H_x x + H_v v \leq r_1^- h^+$. Therefore, because $r_2^+h^- \geq r_1^-h^+$, $H_x x + H_v v \leq r_2^+ h^-$. From here it can be concluded that any $(x,v) \in O_\infty^-(r_1^-)$ also belongs to $O_\infty^-(r_2^+)$. Therefore, $O_\infty^-(r_1^-)\subseteq O_\infty^-(r_2^+)$.
Finally, we consider the case where $r_1\leq0$, denoted $r_1^-$, and where $r_2\leq0$, denoted $r_2^-$. From \eqref{eqn:oinf--}, it is true that $H_x x + H_v v \leq r_1^- h^+$. Therefore, because $r_2^-h^+ \geq r_1^-h^+$, $H_x x + H_v v \leq r_2^- h^+$. From here it can be concluded that any $(x,v) \in O_\infty^-(r_1^-)$ also belongs to $O_\infty^-(r_2^-)$. Therefore, $O_\infty^-(r_1^-)\subseteq O_\infty^-(r_2^-)$.\\
The three cases for the second statement of Proposition \ref{subset}, regarding $r_2\leq r_1$, can be proven similarly.

% Let $r_2\geq r_1$, and $(x,v)\in O_\infty^-(r_1)$. From \eqref{eqn:oinf-}, it is true that $H_x x + H_v v \leq r_1 h$. Therefore, because $r_2\geq r_1$, $H_x x + H_v v \leq r_2 h$. From here it can be concluded that any $(x,v) \in O_\infty^-(r_1)$ also belongs to $O_\infty^-(r_2)$. Therefore, $O_\infty^-(r_1)\subseteq O_\infty^-(r_2)$. The second statement in the theorem can be proved similarly.
%$O_\infty^-(r_1)\subseteq O_\infty^-(r_2)$ because $\forall (x,v) \in  .   
%Let $r_2\geq r_1$, and let $\{(x_1,v_1)\in O_{\infty}(r_1)\}$ and $\{(x_2,v_2)\in O_{\infty}(r_2)\}$. From Theorem \ref{thm3} we can express $O_{\infty}(r_2)$ in terms of $O_{\infty}(r_1)$ as so: $O_{\infty}(r_2)=\gamma O_{\infty}(r_1)$ where $\gamma=\frac{r_2}{r_1}\geq1$. 
\end{proof}
Note that $O_\infty^-(r_1)\not\subseteq O_\infty^-(r_2)$ if $r_2 < r_1$, and $O_\infty^+(r_1)\not\subseteq O_\infty^+(r_2)$ if $r_2 > r_1$. This result implies that, while the dynamic MAS is positively invariant for a fixed $r$, it may not be positively invariant if $r$ varies over time (conditions for positive invariance under time-varying $r$ are given in Proposition \ref{subset}). The implication of this in terms of constraint management will be discussed in the next subsection. The above two propositions provide the answer to the second question raised in the Introduction around the geometric properties of the dynamic MAS. %if $r$ changes

\subsection{Computational aspects and properties of RG-DC}
\label{RG}

To implement the RG-DC, %The $h^+$ and $h^-$ matrices are scaled by $r(t)$ % the overshoot mitigation constraints are initialized as $y_{tr}\leq 1$ and $y\leq -1$ %scales the $h^-$ and $h^+$ vectors by $r(t)$ to
the values of $y_{tr}(t)$ and $r(t)$ are used at every timestep to determine the appropriate MAS from Table~ \ref{tab:4_combinations}. This MAS is then employed in \eqref{eq:kappa_for_RG} to calculate $\kappa$. We denote the resulting solution by $\kappa_{tr}$. If static constraints are also imposed on the system, we compute \eqref{eq:kappa_for_RG} separately with $O_{\infty,st}$, yielding $\kappa_{st}$. The RG-DC then chooses the minimum of $\{\kappa_{tr},\kappa_{st}\}$ and applies the solution to \eqref{eqn:v[t]} to compute $v(t)$.    %The RG-DC can now be implemented using Algorithm \ref{alg:1}  

 As discussed in the previous subsection, the dynamic MAS may or may not be positively invariant if $r$ changes in real-time (conditions for positive invariance were provided in Proposition \ref{subset}). This implies that the LP in \eqref{eq:kappa_for_RG} may be infeasible, which means that   $\kappa_{tr}$ may not exist. 
%In these situations, we manually set $\kappa_{tr}$ to either $0$ or $1$ as explained next. %as explained in Section \ref{MAS}, with the support of Theorem \ref{thm3} 
The traditional reference governor handles this situation by forcing $\kappa$ to be 0 (i.e., $v(t)=v(t-1)$). %However, for the RG-DC, $\kappa = 0$ would mean that $v$ might get stuck as $r(t)$ changes, which is undesirable.
RG-DC handles infeasibilities in the same manner. Specifically, if at the current timestep the LP in \eqref{eq:kappa_for_RG} is infeasible, we set $\kappa = 0$. %Our solution is to use $\kappa = 1$ (i.e., $v(t)=r(t)$) if $v(t-1)$ and $y(t)$ are on opposite sides of $r(t)$ and $\kappa = 0$ (i.e., $v(t)=v(t-1)$) if $v(t-1)$ and $y(t)$ are on the same side as $r(t)$.
In such cases, overshoot is not preventable.
However, we maintain $\kappa_{tr} \in [0,1]$ in our RG-DC formulation to assure stability at the expense of overshoot mitigation performance. % Setting $\kappa=1$ is analogous to temporarily disabling the overshoot-mitigating component of the governor, resulting in the natural closed-loop system response.
 We demonstrate a scenario where the RG-DG forces $\kappa=0$ in Section \ref{sim}.

We discuss the stability of RG-DC in Theorem \ref{stability} below. 
%This section describes the computational aspects and implementation of the RG-DC. Additionally, we prove that the implementation of switching logic in the RG-DC does not destabilize the feedback system. %From \eqref{eqn:v[t]}, it is clear that $v(t)$ monotonically converges to $r(t)$ as it can be described by a convex combination of $r(t)$ and $v(t-1)$. This implies stability of the system.  
\begin{theorem}
\label{stability}
The RG-DC loop is BIBO stable, and for a constant $r$, $v$ converges to a constant.
\end{theorem}
\begin{proof}
From \eqref{eqn:v[t]}, and with $\kappa \in [0,1]$, it follows that $v(t)$ is a convex combination of $r(t)$ and $v(t-1)$, both of which are bounded. Therefore, $v(t)$ is bounded. Boundedness of $v(t)$ and asymptotic stability of \eqref{eqn:RGDC_system} imply BIBO stability of the system. Furthermore, $v(t)$ forms a monotonic sequence bounded by  $r$, which implies convergence.
\end{proof}
Note that this result is similar to the stability result of the standard RG. However, we present it formally to reinforce the claim that, like the RG, the RG-DC is BIBO stable.

%Next we show that at the time of switching, positive invariance of $O_\infty$ holds.

The computational complexity of the RG-DC is similar to that of the standard RG with an additional constraint on the tracking output. Note that the additional logic introduced to determine the MAS characterization from Table \ref{tab:4_combinations} is negligible when compared to the processing times associated with the calculation of $\kappa$ in \eqref{eq:kappa_for_RG}. The RG-DC is also comparable to the RG in terms of memory requirements.

Finally, note that external disturbances, model uncertainty, and plant variability can be naturally incorporated in the RG-DC framework. This is done, similar to standard RG, by ``robustifying" (i.e., shrinking) the MAS using the ideas from Pontryagin subtraction (P-subtraction)  \cite{Gilbert_Kolmanovsky_1995} and polytopic uncertainties \cite{Pluymers_2005}. We will show an example of this in the next section.

The above analyses provide complete answers to questions 3 -- 6 raised in the Introduction around the properties of RG-DC.

We now present the RG-DC algorithm (see Algorithm \ref{alg:RGDC}), which can be used to enforce overshoot mitigation constraints using $O_{\infty,tr}$ and static constraints using $O_{\infty,st}$. In preparation for Algorithm \ref{alg:RGDC}, assume that the two representations of the dynamic MAS, namely $H_{x}^-,H_{v}^-,h^-$ and $H_{x}^+,H_{v}^+,h^+$, have been calculated. Let $H_{x,tr},H_{v,tr},h_{tr}$ be the representation with the larger number of rows, where $h_{tr}$ is a vector of all $1$s. Additionally, let $H_{x,st},H_{v,st},h_{st}$  denote the matrices that define $O_{\infty,st}$. The RG-DC algorithm is as follows:  

\begin{algorithm}[h!]
    \caption{RG-DC}
    \vspace{.03in}
    \textbf{Inputs:} \\$y(t),r(t),x(t),v(t-1),H_{x,tr},H_{v,tr},h_{tr},H_{x,st},H_{v,st},h_{st},\epsilon$ \\\textbf{Output:} \\$v(t)$
    \vspace{.03in}
    \hrule
    \vspace{.03in}
  \begin{algorithmic}[1]
  %\SetAlgoNoLine
    \IF{$y_{tr}(t) \leq r(t)$}
        \IF{$r(t) > 0$}
            \STATE first row of $h_{tr} = (1-\epsilon)$
        \ELSE
            \STATE first row of $h_{tr} = (1+\epsilon)$
        \ENDIF
        \FOR{$j$ = each row in $0_{\infty,tr}$}
            \STATE $n=h_{tr}(j)r(t)-H_{x,tr}(j) x(t)-H_{v,tr}(j) v(t-1)$
            \STATE $d=H_{v,tr}(j)(r(t)-v(t-1))$
            \STATE $\kappa(j)=\mathrm{kappa}(n,d)$
        \ENDFOR
    \ELSE
        \IF{$r(t) > 0$}
            \STATE first row of $h_{tr} = (1+\epsilon)$
        \ELSE
            \STATE first row of $h_{tr} = (1-\epsilon)$
        \ENDIF
        \FOR{$j$ = each row in $0_{\infty,tr}$}
            \STATE $n=-h_{tr}(j)r(t)+H_{x,tr}(j) x(t)+H_{v,tr}(j) v(t-1)$
            \STATE $d=-H_{v,tr}(j)(r(t)-v(t-1))$
            \STATE $\kappa(j)=\mathrm{kappa}(n,d)$
        \ENDFOR
    \ENDIF
    \STATE $\kappa_{tr} = \mathrm{min}(\kappa)$
    \IF{there are any static constraints}
        \STATE use standard RG algorithm with $O_{\infty,st}$ to obtain $\kappa_{st}$
    \ELSE
        \STATE $\kappa_{st}=1$
    \ENDIF
    \STATE $\kappa^*=\mathrm{min}(\kappa_{tr},\kappa_{st})$
    \STATE $v(t) = v(t-1)+\kappa^*(r(t)-v(t-1))$
    \end{algorithmic}
    %\SetAlgoNoLine
    %\SetKwFunction{FMain}{kappa}
    %\SetKwProg{Fn}{Function}{:}{}
    \vspace{.03in}
    \hrule
    \vspace{.05in}
    \textbf{function} $\mathrm{kappa}(n,d)$
    \vspace{.03in}
    %\Fn{\FMain{$n,d$}}{
    \begin{algorithmic}[1]
    \IF{$n > 0$}
        \IF{$d>0$}
            \STATE $\kappa=\mathrm{min}(n/d,1)$
        \ELSE
            \STATE $\kappa=1$
        \ENDIF
    \ELSE
        \STATE $\kappa=0$
    \ENDIF
    \STATE \textbf{return} $\kappa$
    \end{algorithmic}
    \vspace{.03in}
    %}
    
  \textbf{end function}
  \vspace{.03in}
    %\caption{RG-DC \\Inputs: \\$y(t),r(t),x(t),v(t-1),H_{x,tr},H_{v,tr},h_{tr},H_{x,st},H_{v,st},h_{st},\epsilon$ \\Output: \\$v(t)$}
    \label{alg:RGDC}
\end{algorithm}

\section{Illustrative Examples}
\label{sim}
\subsection{System model}
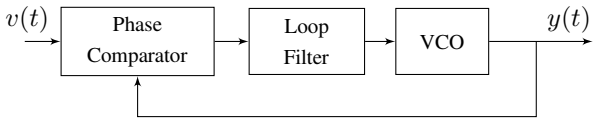
\begin{figure}[ht!]
\centering
\begin{tikzpicture}[auto, node distance=1.5cm,>=latex']
\node [input, name=rinput] (rinput) {};
\node [block, right of=rinput,text width=1.8cm,align=center] (comparator) {{\footnotesize Phase Comparator}};
\node [block, right of=comparator,node distance=2.25cm,text width=1.3cm,align=center] (filter) {{\footnotesize Loop Filter}};
\node [block, right of=filter, node distance=1.8cm, text width=1cm,align=center] (vco) {{\footnotesize VCO}};
\node [tmp, right of=vco, node distance=1.25cm] (feedback) {};
\node [output, right of=vco, node distance=2cm] (output) {};
\node [tmp, below of=comparator,node distance=1cm] (tmp1){$s$};
\draw [->] (rinput) -- node{\hspace{-0.4cm}$v(t)$} (comparator);
% \draw [->] (sum1) -- (controller);
\draw [->] (comparator) -- node [name=v1]{}(filter);
\draw [->] (filter) -- node [name=v2]{}(vco);
\draw [->] (vco) -- node [name=y] {\hspace{0.75cm}$y(t)$}(output);
\draw [->] (feedback)|-(tmp1)-| node[pos=0] {} (comparator);
% \draw [->] (extra)--(sum2);
% \draw [->] ($(0,1.5cm)+(extra)$)node[above]{$d_{\beta 2}$} -- (extra);
\end{tikzpicture}
\caption{Simple analog PLL system.} 
\label{fig:PLLblock}
\end{figure}

Consider the analog phase locked loop (PLL) system shown in Fig. \ref{fig:PLLblock}, which is comprised of a phase comparator, a loop filter, and a voltage controlled oscillator (VCO). The transfer function of the closed-loop PLL system around a nominal operating point is as follows \cite{li2000introduction}:
\begin{equation}
\begin{aligned}
   %H_{lp} = \frac{G_{lp}}{G_{lp}+s},\ 
    %H_{VCO} = \frac{G_{VCO}}{s},\
    H_{PLL} = \frac{G_{lp}G_{VCO}}{s^2+G_{lp}s+G_{lp}G_{VCO}}
\end{aligned}
\label{eqn:PLL}
\end{equation}
where $G_{lp}$ is the  loop filter parameter and $G_{VCO}$ is the VCO gain.
Note that the closed-loop system has a DC gain of $1$ and perfect steady-state tracking of step commands. For the simulation, $G_{lp}$ was chosen to be $100$, and $G_{VCO}$ was chosen to be $2G_{lp}$ to yield an  underdamped system with damping ratio $\zeta=0.35$.  By selecting states as $x_1=y$ and $x_2=\dot{y}$, a zero order hold discretization of the system with a sample time of $T_s=1\times10^{-4}$ seconds is used to obtain the discrete state-space model of the closed-loop system.% yields the following discrete-time state space model: %our system with a zero order hold obtain a continuous-time state-space model with system matrices:
% \begin{equation*}
% % \begin{aligned}{
% A = 
% \begin{bmatrix}
% 0&1\\
% -G_{lp}G_{vco}&-G_{lp}\\
% \end{bmatrix},
% B = 
% \begin{bmatrix}
% 0\\
% G_{lp}G_{vco}\\
% \end{bmatrix}
% C = 
% \begin{bmatrix}
% 1&0\\
% 0&1\\
% \end{bmatrix},
% D =
% \begin{bmatrix}
% 0\\
% 0\\
% \end{bmatrix}
% }
% % \end{aligned}
% \end{equation*}
% and $C=I, D=0$. For the simulation, the loop filter gain $G_{lp}$ was chosen to be $100$ therefore attenuating signals above $16$ Hz. The VCO gain $G_{vco}$ was chosen to be $2G_{lp}$ to yield an  underdamped system with damping ratio $\zeta=0.35$. With these parameters, the discrete-time state space matrices of the system are realized as:
% \begin{equation*}
% A = 
% \begin{bmatrix}
% 1.00&9.95\times10^{-5}\\
% -1.99&0.99
% \end{bmatrix},\ 
% B = 
% \begin{bmatrix}
% 9.97\times10^{-5}\\
% 1.99
% \end{bmatrix},\ C=I,\ D=0\\
% \end{equation*}

Constraints are imposed on both outputs of the system. The dynamic constraint is applied to the tracking output $y_{tr} \triangleq y_1$ and a slew-rate limiting constraint ($-100\leq y_2\leq 100$) is applied to the constrained output $y_{st} \triangleq y_2$. The static and dynamic maximal admissible sets are constructed as discussed in Sections \ref{review} and \ref{DCM}. The resulting polyhedra, $O_\infty^-(r^+)$ and $O_\infty^-(r^-)$, both have admissibility indices of $342$ (the representations happen to be  the same for this example). Additionally, the admissibility index of $O_{\infty,st}$ is $130$.
%The above constraints, together with $y_1\leq1$, form the MAS presented in Fig. \ref{fig:MAS_Slew}. Note that as the control command approaches the tracking output constraint, MAS is generally tightened.
% \begin{figure}[ht]
%     \centering
%     \includegraphics[width=3.4in]{Figures/JP_MAS.eps}
%     \caption{MAS polytope slices with slew-rate limit}
%     \label{fig:MAS_Slew}
% \end{figure}
%\newpage
\subsection{Response Evaluation}
Fig. \ref{fig:Comparison} shows the improved response characteristics of the governed system compared to the ungoverned system. Note that overshoot was completely eliminated without making any modifications to the PLL. Hence, the RG-DC is especially effective in overshoot mitigation of systems with inner loop controllers that cannot be tuned or adjusted (i.e `black box' systems), which is true for many off-the-shelf PLLs.     %Additionally, a critically damped system is included in the plot for system speed comparison with zero overshoot. The critically damped system was derived from the underdamped PLL system simply by decreasing the VCO gain from $2G_{lp}$ to $0.25G_{lp}$. Therefore the critically damped system still filters control signal frequencies above $16$ Hz. %By evaluating all three response traces, the data in Table \ref{tab:parameters} can be obtained. 
\begin{figure}[ht]
    \centering
    \includegraphics[width = \columnwidth]{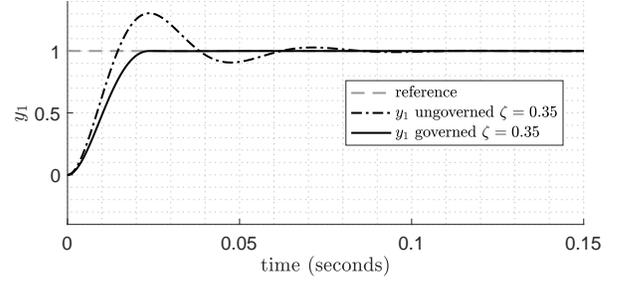}
        %\vspace{-0.2cm}
    \caption{Comparison between step responses of the governed and ungoverned systems. No slew-rate limit is applied to the governed system in this simulation.}
    \label{fig:Comparison}
\end{figure}
% \renewcommand{\arraystretch}{1.2}
% \begin{table}[h!]
%     \centering
%     \begin{tabular}{|r|r|r|}
%         \hline
%         response & rise time ($\times 10^{-2}$s) & settling time ($\times 10^{-2}$s) \\
%         \hline
%         governed $\zeta=0.35$ & $1.4$ & $2.4$\\
%         \hline
%         ungoverned $\zeta=0.35$ & $1.0$ & $8.2$\\
%         \hline
%         ungoverned $\zeta=1.00$ & $6.7$ & $13.3$\\
%         \hline
%     \end{tabular}
%     \vspace{.1in}
%     \caption{Response performance parameters}
%     \label{tab:parameters}
% \end{table}
% \renewcommand{\arraystretch}{1.1}

\begin{figure}
%\vspace{-0.3cm}
    \centering
        %\vspace{-0.2cm}
    \includegraphics[width = \columnwidth]{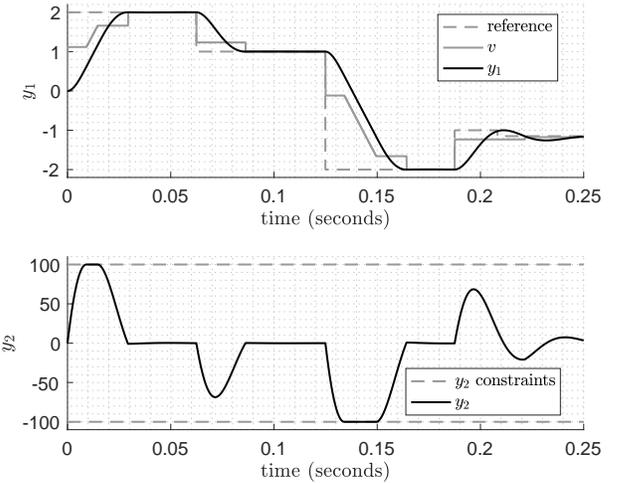}
        %\vspace{-0.2cm}
    \caption{Governed system response to multiple steps inputs (slew-rate limit = 100)}
    \label{fig:4}
\end{figure}

A simulation of the RG-DC operating on step signals is presented in Fig. \ref{fig:4}. Notice that overshoot is mitigated for all but the last step at $t=0.208$ seconds. In this case, the reference changes quickly so that $(x(t),v(t-1))$ does not belong to the new MAS, which means constraint violation is not preventable. Hence, $\kappa$ has been set to $0$. Note that we maintain convergence to the reference at the sacrifice of reduced overshoot mitigation performance. %supporting the theory from Section \ref{DCM} and demonstrating the effectiveness of the RG-DC algorithm.

\subsection{Robustness}
To test robustness under model uncertainty, we treat the VCO gain, $G_{VCO}$, as an unknown. We suppose, however, that $G_{VCO}$ is bounded as follows: $160\leq G_{VCO}\leq 240$. We compute a robust MAS for this  system  using Algorithm $1$ from \cite{Pluymers_2005}. %For the simulation the VCO gain of the system was chosen to be 7540, falling directly in the middle of the uncertainty bounds.
%We then compute two MASs for each vertex of the uncertainty polytope (i.e., one for 160 and one for 240) and take their intersection to obtain the robust MAS. Convexity ensures that the constraints are satisfied if the state and input pair belongs to this MAS. 
Fig. \ref{fig:Plytope_Slices} compares the robust MAS with a standard MAS generated with the nominal model parameter $G_{VCO}=200$. From the figure, it is evident that the introduction of model uncertainty results in a more conservative MAS. Upon simulation of the governed system with the robust MAS, we see in Fig. \ref{fig:Uncertainty_Plot} that the constraints are not violated for systems at the vertices of system uncertainty (the values of $G_{VCO}$ used for the simulations are shown in the figure titles).
\begin{figure}
    \centering
    \includegraphics[width= \columnwidth]{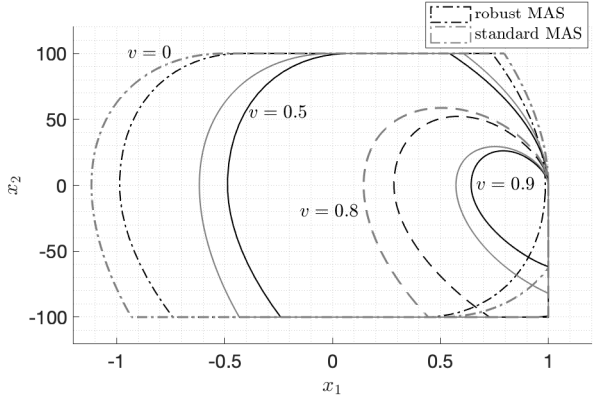}
           %\vspace{-0.2cm}

\caption{Slices from the robust and standard MASs at various values of $v$. The dynamic constraint for this plot is $y_1\leq 1$ and the slew-rate limit is 100.}

    \label{fig:Plytope_Slices}
\end{figure}
\begin{figure}
    \centering
    \includegraphics[width = \columnwidth]{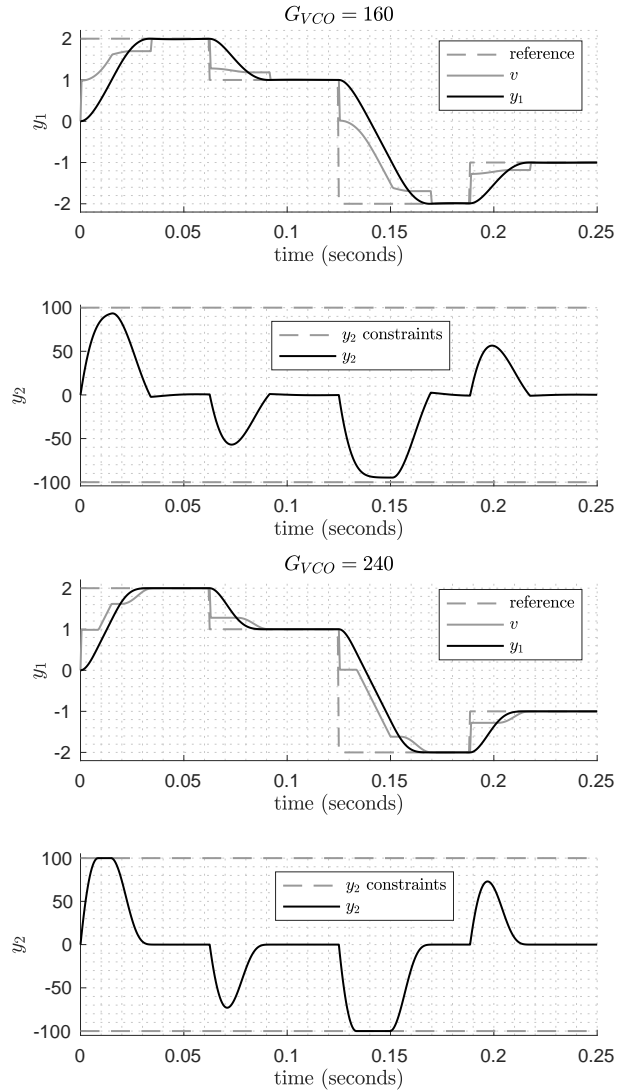}
        %\vspace{-0.4cm}
    \caption{Governed uncertain system responses to multiple step inputs (slew-rate limit = $100$). Note that the top two sub-figures and the bottom two sub-figures correspond to two different realizations of the system uncertainty shown in the plot titles.}
        %\vspace{-0.1cm}
    \label{fig:Uncertainty_Plot}
\end{figure}

\subsection{RG-DC as a nonlinear filter}
Finally, we present an interesting experiment, which led to a thought-provoking observation regarding the frequency response of the governed PLL system, which we discuss next.

%By construction, the RG-DC is inherently nonlinear. %, which raises the question, how do we represent and compare the frequency response of the governed nonlinear system to the ungoverned linear system?
In \cite{Nonlinear_Bode}, it is shown that nonlinear systems can be analyzed using frequency domain techniques if they satisfy the ``convergence" property. Essentially, a system is defined to be convergent if, akin to a linear system, its response  converges to the forced response, regardless of the initial conditions. As argued in \cite{Nonlinear_Bode}, nonlinear convergent systems can be analyzed using the \textit{nonlinear Bode magnitude plot}, which is a proper extension of the traditional Bode magnitude plot for linear systems. However, unlike the linear Bode plot, which is only a function of the frequency of the input sinusoid, the nonlinear Bode plot is generally a function of both the frequency and amplitude of the input. %Thus, given a bounded reference signal, the output of a nonlinear convergent system will converge independent of the initial conditions. 

In our case, it can be shown that the overall system with the RG-DC governing the input is indeed a convergent system (see Fig. \ref{fig:Convergence} for graphical argument). Furthermore, as we prove in Theorem \ref{thm:scaling}, the RG-DC is such that the nonlinear Bode plot has no dependence on the amplitude of the input because  the system satisfies the homogeneity condition, similar to a linear system. %Furthermore we prove in Theorem \ref{thm:scaling} that any RG-DC governed system is ``linear" in the sense that scaling the reference signal and initial conditions results in an equal scaling of the modified reference and therefore an equal scaling of the system output (because the closed-loop system is linear).
Thus, we adopt the methods in \cite{Nonlinear_Bode} to generate a nonlinear Bode magnitude plot of the governed system as a function of the input frequency only. This plot is presented in Fig. \ref{fig:Bode}, which also shows the Bode magnitude plot of the ungoverned PLL system \eqref{eqn:PLL} as comparison. The other plots labeled ``$2^{\mathrm{nd}}$ order system" and ``$12^{\mathrm{th}}$ order system" will be explained later.
%{\color{red}Additionally we include the frequency response of a system described by the continuous-time transfer function: $\frac{266^2}{(s+266)^2}$, in order to compare the frequency response of the governed PLL system to that of a linear system of the same order as the PLL system, unity DC gain, no resonance, and $-3$ dB bandwidth equal to that of the governed system.} 
Details on how each plot was generated can be found in the caption of Fig. \ref{fig:Bode}. Upon inspection of Fig. \ref{fig:Bode}, it appears that the resonant peak inherent in the Bode magnitude plot of the underdamped closed loop PLL system is completely eliminated with the implementation of the RG-DC.

\begin{figure}
    \centering
        %\vspace{-0.4cm}
    \includegraphics[width = \columnwidth]{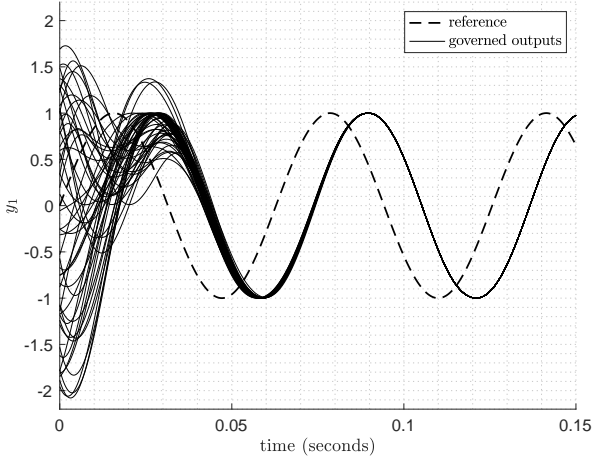}
    %\vspace{-0.4cm}
    \caption{Demonstration of convergence via simulation of the governed PLL system (no slew-rate limit) at $50$ jointly uniformly distributed random initial conditions $\left(x_0,v_0\right)$. Initial condition ranges: $x_{0_1} \in [-2,2],\ x_{0_2} \in [-200,200],\ v_0 \in [-1,1]$. The reference $r(t)$ is a sinusoid with frequency $100$ rad/s. Note that overshoot mitigation constraints for some initial conditions were not satisfied because the initial conditions did not belonged to MAS.}
    \label{fig:Convergence}
\end{figure}

\begin{figure}
    \centering
        %\vspace{-0.4cm}
    \includegraphics[width = \columnwidth]{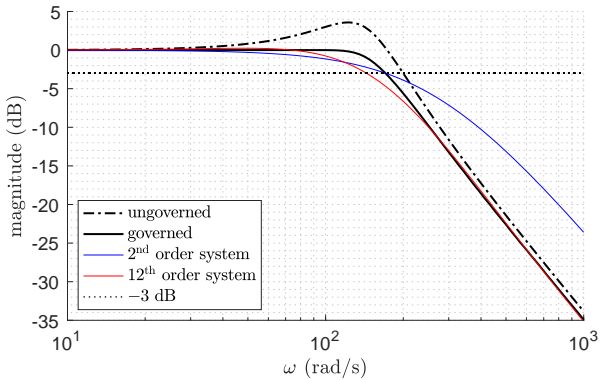}
    %\vspace{-0.4cm}
    \caption{Bode magnitude plot (ungoverned PLL system) and nonlinear Bode magnitude plot (governed PLL system, no slew-rate limit). In addition, the Bode magnitude plots of the $2^{\mathrm{nd}}$ order and $12^{\mathrm{th}}$ order systems are shown. The nonlinear Bode magnitude plot was generated by simulating governed system responses with sinusoidal references of amplitude $1$ at $100$ different frequencies that were logarithmically equally spaced ranging from $10$ rad/s to $1,000$ rad/s. The supremum norm of the outputs were measured at steady-state and converted to dB. The ungoverned PLL system, $2^{\mathrm{nd}}$ order system, and $12^{\mathrm{th}}$ order system Bode magnitude plots were generated using the standard linear systems approach applied to the respective linear system models.}
    \label{fig:Bode}
\end{figure}

 These results suggest that RG-DC could potentially be used in conjunction with a resonant low-pass filter, giving the ability to eliminate the resonant behavior without greatly affecting the cutoff frequency or the attenuation properties beyond the cutoff frequency. In other words, the RG-DC may be thought of as a ``nonlinear" filter with anti-resonance properties.
 
We highlight the fact that the resonant peak in the Bode plot of the ungoverned system shown in Fig. \ref{fig:Bode} is caused by underdamped poles in \eqref{eqn:PLL}. Therefore, a natural solution to remove the resonant peak is by using plant inversion, i.e., replacing the RG-DC in Fig. \ref{fig: Governor scheme} with an inverse model that cancels the underdamped poles of the closed-loop system with a pair of complex conjugate  zeros. Since the resulting inverse model would be improper, additional (overdamped) poles must be added to obtain a proper transfer function. The series connection of the inverse model and the closed-loop system yields an equivalent transfer function with no resonance peaks. This raises the following question: how does the nonlinear Bode plot from the governed PLL system compare with the linear Bode  plot of this equivalent system? To have a fair comparison, we introduce two choices for the equivalent systems, both with DC gain equal to $1$ and relative degree equal to $2$ to maintain the $-40$ dB/decade roll-off of the governed PLL system. The first system, of order $2$, was designed such that the $-3$ dB bandwidth was equal to that of the governed PLL system (based off the nonlinear Bode magnitude plot), whereas the second system, of order $12$, was designed to minimize the root-mean-square error relative to the nonlinear Bode magnitude plot. The results are shown in Fig. \ref{fig:Bode}. Note that the frequency response of the governed PLL system has a faster transition from $0$  to $-40$ dB/decade compared to the frequency response of the $2^{\mathrm{nd}}$ order system. %Furthermore, the attenuation properties of the governed PLL system at frequencies above the cutoff frequency is much greater for the governed PLL system and more closely resembles the behavior of the ungoverned PLL system. 
Furthermore, the $12^{\mathrm{th}}$ order Bode magnitude plot matches the nonlinear Bode magnitude plot except for the frequency range from $100$ rad/s to $200$ rad/s  where the magnitude of the $12^{\mathrm{th}}$ order frequency response is less than the magnitude of the governed frequency response. We thus make the mild conclusion that the frequency response provided by the governed PLL system is not attainable by a low-order linear system (of degree less than $12$), which shows that an RG-DC governed resonant low-pass filter does indeed produce a novel frequency response.
 
 %with continuous-time transfer functions: $\frac{266^2}{(s+266)^2}$ and $\frac{(s+70)(s+160)(s+205)(s+215)(s+225)(s+235)(s+245)(s+250)(s+260)(s+265)}{5.9678\times10^{-5}(s+195)^{12}}$

We now prove the homogeneity property of the overall system with the RG-DC governing the reference, as alluded to above. In preparation for Theorem \ref{thm:scaling}, we introduce the following notation.
Let the governed output of system \eqref{eqn:RGDC_system} be $y_{tr}\left(t,r(t),(x_0,v_0)\right)$,
where $r(t)$ is the reference signal that is applied to the system depicted in Fig. \ref{fig:RGblock} with initial conditions $(x_0,v_0):=(x(0),v(-1))$ belonging to MAS. The following theorem holds. 

\begin{theorem}
\label{thm:scaling}
Suppose $p=0$ in system  \eqref{eqn:RGDC_system}, so that $y_{tr}(t)$ is the only output governed by RG-DC. Then $y_{tr}\left(t,\alpha r(t),\alpha(x_0,v_0)\right) = \alpha y_{tr}\left(t,r(t),(x_0,v_0)\right)$, $\forall \alpha \in \mathbb{R}^+$.
%and $y(t) = f\left(r(t),\left(x_0,v_0\right)\right)$, then $\alpha y(t) = f\left(\alpha r(t),\alpha \left(x_0,v_0\right)\right)$.  
\end{theorem}
\begin{proof}
We prove the homogeneity condition of the RG-DC from $r$ to $v$ by principal of induction. The homogeneity condition from $r$ to $y_{tr}$ then follows from the fact that the closed-loop system \eqref{eqn:RGDC_system} from $v$ to $y_{tr}$ is linear. 
% \begin{equation}\label{eq:kappa_for_RGt}
% \begin{aligned}
% &\underset{\kappa\in [0,1]}{\text{maximize}}
% & & \mathrm{\kappa} \\
% & \hspace{10pt} \text{s.t.}
% & & v(t)=v(t-1)+\kappa\left(r(t)-v(t-1)\right)\\
% &&&\left(x(t),\ v(t)\right) \in O_{\infty,tr}\left(r(t),C_{tr}x(t)\right)
% \end{aligned}
% \end{equation}
% where $x(1)=Ax(t-1)+Bv(t-1)$. Consider the specific time instance $t=1$ and let the initial conditions be $(x_0,v_0)=(x(0),v(0))$. The optimization problem \eqref{eq:kappa_for_RGt} becomes:
% \begin{equation}\label{eq:kappa_for_RGt1}
% \begin{aligned}
% &\underset{\kappa\in [0,1]}{\text{maximize}}
% & & \mathrm{\kappa} \\
% & \hspace{10pt} \text{s.t.}
% & & v(1)=v_0+\kappa\left(r(1)-v_0\right)\\
% &&&\left(x(1),\ v(1)\right) \in O_{\infty,tr}\left(r(1),(C_{tr}x(1))\right).
% \end{aligned}
% \end{equation}
%where $x(t)=Ax(t-1)+Bv(t-1)$. 

We first establish the base case of the inductive argument, where we prove that scaling the initial conditions, $(x_0,v_0)=(x(0),v(-1))$, and the reference, $r(0)$, by $\alpha$ (written in short by $(x_0,v_0)\rightarrow \alpha (x_0,v_0)$, $r(0) \rightarrow \alpha r(0)$), scales the next iterate by $\alpha$:  $(x(1),v(0))\rightarrow \alpha (x(1),v(0))$. To show this, consider the RG algorithm from \eqref{eq:kappa_for_RG} at time $t=0$ with the change of variables: $x(0)\rightarrow \alpha x(0)$, $v(-1)\rightarrow \alpha v(-1)$, and $r(0)\rightarrow \alpha r(0)$. Finally, let $\widetilde{O}_\infty$ from \eqref{eq:kappa_for_RG} be $O_{\infty,tr}\left(\alpha r(0),\alpha y_{tr}(0)\right)=O_{\infty,tr}\left(\alpha r(0),\alpha C_{tr}x_0\right)$. The optimization problem becomes:
\begin{equation}\label{eq:base_case}
\begin{aligned}
&\underset{\kappa\in [0,1]}{\text{maximize}}
& & \mathrm{\kappa} \\
& \hspace{10pt} \text{s.t.}
& & v(0)=\alpha v_0+\kappa\left(\alpha r(0)-\alpha v_0\right)\\
&&&\left(\alpha x_0,\ v(0)\right) \in O_{\infty,tr}\left(\alpha r(0),\alpha C_{tr}x_0\right)
\end{aligned}
\end{equation}
% Now consider the following change of variables to \eqref{eq:base_case}: $r(0)\rightarrow \alpha r(0)$, $x_0\rightarrow \alpha x_0$, and $v_0\rightarrow \alpha v_0$. Optimization problem \eqref{eq:base_case} becomes:
% \begin{equation}\label{eq:base_case_alpha}
% \begin{aligned}
% &\underset{\kappa\in [0,1]}{\text{maximize}}
% & & \mathrm{\kappa} \\
% & \hspace{10pt} \text{s.t.}
% & & v(0)=\alpha v_0+\kappa\left(\alpha r(0)-\alpha v_0\right)\\
% &&&\left(\alpha x_0,\ v(0)\right) \in O_{\infty,tr}\left(\alpha r(0),\alpha C_{tr}x_0\right)
% \end{aligned}
% \end{equation}
Recall from Table \ref{tab:4_combinations}, that the relationship between $y_{tr}(t)$ and $r(t)$ determines which of the four cases of dynamic MAS is used in the RG algorithm at timestep $t$. Furthermore, note that scaling $r(t)$ and $y_{tr}(t)$ by $\alpha$ does not change which case is appropriate. This implies that $O_{\infty,tr}\left(r(t), C_{tr}x(t)\right)$ and $O_{\infty,tr}\left(\alpha r(t), \alpha C_{tr}x(t)\right)$, at any instance in time, both belong to the same case from Table \ref{tab:4_combinations}. By Proposition \ref{thm3} $i)$, it then follows that $O_{\infty,tr}\left(\alpha r(0), \alpha C_{tr}x_0\right)=\alpha O_{\infty,tr}\left(r(0), C_{tr}x_0\right)$, and we can conclude that the constraints of optimization problem \eqref{eq:base_case} are unaffected by $\alpha$. Furthermore, noting that the cost function of \eqref{eq:base_case} also does not depend on $\alpha$, it follows that optimization problem \eqref{eq:base_case} results in the same optimizer $\kappa^*$ at time $t=0$ regardless of $\alpha$. From here, it follows that the modified reference from \eqref{eq:base_case} is $\alpha v(0)$. Furthermore, at time $t=1$ we can conclude that $\alpha x(1)=A\alpha x_0+B\alpha v(0)$. Note that this base case also holds when there is no solution to the optimization problem because $\kappa^*=0$ at time $t=0$ means that $\alpha v(0)= \alpha v_0$.

We now present the induction step, where we prove that scaling the parameters, $x(t)$, $v(t-1)$, and $r(t)$, by $\alpha$, gives the following result:  $v(t)\rightarrow \alpha v(t)$ and $x(t+1)\rightarrow \alpha x(t+1)$. Consider the RG algorithm from \eqref{eq:kappa_for_RG} with the change of variables: $x(t)\rightarrow \alpha x(t)$, $v(t-1)\rightarrow \alpha v(t-1)$, and $r(t)\rightarrow \alpha r(t)$. Finally, let $\widetilde{O}_\infty$ from \eqref{eq:kappa_for_RG} be $O_{\infty,tr}\left(\alpha r(t),\alpha y_{tr}(t)\right)=O_{\infty,tr}\left(\alpha r(t),\alpha C_{tr}x(t)\right)$. The optimization problem becomes: 
\begin{equation}\label{eq:Induction_Step}
\begin{aligned}
&\underset{\kappa\in [0,1]}{\text{maximize}}
& & \mathrm{\kappa} \\
& \hspace{10pt} \text{s.t.}
& & v(t)=\alpha v(t-1)+\kappa\left(\alpha r(t)-\alpha v(t-1)\right)\\
&&&\left(\alpha x(t),\ v(t)\right) \in O_{\infty,tr}\left(\alpha r(t),\alpha C_{tr}x(t)\right)
\end{aligned}
\end{equation}
Again, note that because $O_{\infty,tr}\left(\alpha r(t),\alpha C_{tr}x(t)\right)=\alpha O_{\infty,tr}\left(r(t), C_{tr}x(t)\right)$, the constraints of optimization problem \eqref{eq:Induction_Step} are independent of $\alpha$. Furthermore, noting that the cost function of \eqref{eq:Induction_Step} also does not depend on $\alpha$, it follows that \eqref{eq:Induction_Step} results in the same optimizer $\kappa^*$ at time $t$ regardless of $\alpha$. From here, it follows that the modified reference from \eqref{eq:Induction_Step} is $\alpha v(t)$. 
Furthermore, at time $t+1$ we can conclude that $\alpha x(t+1)=A\alpha x(t)+B\alpha v(t-1)$. Note that this induction step also holds when there is no solution to the optimization problem because $\kappa^*=0$ at time $t$ means that $\alpha v(t)= \alpha v(t-1)$. Considering the above logic, by the principal of induction, the RG-DC satisfies the homogeneity condition (from $r$ to $v$) and we can conclude that, because the closed-loop system is linear (from $v$ to $y_{tr}$), the entire governed system (from $r$ to $y_{tr}$) satisfies the homogeneity condition. This completes the proof. 
\end{proof}

An additional argument can be made that further strengthens the validity of the nonlinear Bode magnitude plot. We argue that although the input to the linear closed-loop system ($v(t)$) is not perfectly sinusoidal (due to the governing action of the RG-DC), the output of the system ($y_{tr}$) is ``approximately" sinusoidal (which we have observed in our simulations, see for example Fig. \ref{fig:Convergence}). The reason for this phenomenon can be attributed to the fact that the closed-loop system is of low-pass-filtering nature, which implies that it filters out higher order harmonics of $v(t)$. This argument is similar to the methods used for Describing Functions \cite{Khalil_96}. 

 %designed to minimize phase lag before the cutoff frequency (poles have large imaginary components) to eliminate their inherently high resonance. %Furthermore, the RG-DC can be made robust to model uncertainty which gives it the advantage over feedforward plant inversion. 

\section{Conclusion}

In this paper, an overshoot mitigation control scheme was developed using the reference governor framework. The solution, known as the Reference Governor with Dynamic Constraint (RG-DC), utilizes a dynamic maximal admissible set (MAS) to constrain the tracking output such that overshoot of step inputs is eliminated. The RG-DC loop was proven to be BIBO stable. Additionally, properties of the dynamic MAS were studied and theorems were proven that allow for the RG-DC to operate without recalculation of the matrices that define the dynamic MAS. %upon a change in the constraints. 
While the RG-DC can guarantee overshoot elimination for all step inputs with the proper initial conditions, it may not remove overshoot for a more general time-varying reference $r$. Conditions were provided in the paper under which elimination will be guaranteed for time-varying $r$.

%In summary, RG-DC is a computationally-efficient add-on scheme that successfully tracks a reference signal while mitigating overshoot. 
Future work on the RG-DC and its effect on frequency response are of interest. More specifically, we would like to study the settling time of the system under RG-DC and explore the application of the RG-DC as a nonlinear filter.

 %a reference governor scheme for overshoot mitigation in a tracking control system. The solution, referred to as the Reference Governor with Dynamic Constraint (RG-DC), recasts the overshoot mitigation problem into a constraint management problem. The outcome of this reformulation is a novel maximal admissible set (MAS), which varies in real-time as a function of the reference signal. RG-DC employs this dynamic MAS, as well as a novel switching logic, to compute the input commands to prevent or mitigate overshoot. The properties and the computation of this MAS are studied, and the stability and recursive feasibility of RG-DC are investigated. Simulation results demonstrate the efficacy of the approach, but also highlight its limitations.

%\vspace{-0.2cm}
\section*{Disclaimer}
\balance
%Official contribution of the National Institute of Standards and Technology;} not subject to copyright in the United States. %Certain commercial equipment, instruments, or materials are identified in this paper in order to specify the experimental procedure adequately. Such identification is not intended to imply recommendation or endorsement by the National Institute of Standards and Technology, nor is it intended to imply that the materials or equipment identified are necessarily the best available for the purpose. 
\small
{Portions of this publication and research effort are made possible through the help and support of NIST via cooperative agreement 70NANB19H133. Official contribution of the National Institute of Standards and Technology; not subject to copyright in the United States. Certain commercial products are identified in order to adequately specify the procedure; this does not imply endorsement or recommendation by NIST, nor does it imply that such products are necessarily the best available for the purpose.}
%\vfill
%\vspace{-.05in}
\bibliographystyle{unsrt}
\bibliography{main}
\end{document}